\documentclass[a4paper,USenglish,cleveref,autoref,numberwithinsect,thm-restate,authorcolumns]{lipics-v2019}

\usepackage{booktabs}
\usepackage{amssymb}
\usepackage{amsmath}
\usepackage{graphicx}
\usepackage[font=small,labelfont=bf]{caption}
\usepackage[labelfont=default]{subcaption}
\usepackage{mathtools}
\mathtoolsset{showonlyrefs=true}
 \usepackage{todonotes}
\usepackage{hyperref}
\usepackage{enumerate}
\usepackage{xspace}
\usepackage{cite}

\newcommand{\mktodo}[2][]{}

\newtheorem{observation}[theorem]{Observation}

\graphicspath{{./figures/}}

\newcommand{\rac}[1]{RAC$_{#1}$\xspace}
\newcommand{\arac}{arc-RAC\xspace}

\newcommand{\fqp}{$4$-quasi-planar\xspace}

\newcommand{\ra}[1]{r({#1})} 
\newcommand{\ce}[1]{C({#1})} 

\newcommand{\ch}[1]{ch({#1})} 
\newcommand{\chfin}[1]{ch_{\mathrm{fin}}({#1})}
\newcommand{\ord}[2]{\Pi({#1}; {#2})} 

\newcommand{\cv}[1]{\bar{#1}} 
\newcommand{\img}[1]{{#1}^\circ} 

\newcommand\restr[2]{{
  \left.\kern-\nulldelimiterspace 
  #1 
  \vphantom{\big|} 
  \right|_{#2} 
}}

\title{Drawing Graphs\\ with Circular Arcs and Right-Angle Crossings}
\titlerunning{Drawing Graphs with Circular Arcs and Right-Angle Crossings}

\author{Steven Chaplick}{Maastricht University, the Netherlands \linebreak University of W\"urzburg,  Germany}{s.chaplick@maastrichtuniversity.nl}{https://orcid.org/0000-0003-3501-4608}{}

\author{Henry~F\"orster}{University of T\"ubingen, Germany}{foersth@informatik.uni-tuebingen.de}{https://orcid.org/0000-0002-1441-4189}{}

\author{Myroslav Kryven}{University of W\"urzburg,  Germany}{myroslav.kryven@uni-wuerzburg.de}{}{}

\author{Alexander Wolff}{University of W\"urzburg,  Germany}{}{https://orcid.org/0000-0001-5872-718X}{}

\authorrunning{S.~Chaplick, H.~F\"orster, M.~Kryven, and A.~Wolff}

\Copyright{Steven Chaplick, Henry F\"orster, Myroslav Kryven, and Alexander Wolff}

\ccsdesc[300]{Mathematics of computing~Graphs and surfaces}
\ccsdesc[300]{Mathematics of computing~Combinatoric problems}

\keywords{circular arcs, right-angle crossings, edge density, charging argument}

\category{} 

\relatedversion{} 

\supplement{}

\funding{M.~Kryven acknowledges support from DFG grant WO$\,$758/9-1.}

\nolinenumbers 

\hideLIPIcs

\begin{document}

\maketitle

\begin{abstract}
  In a RAC drawing of a graph, vertices are represented by points in
  the plane, adjacent vertices are connected by line segments, and
  crossings must form right angles.  Graphs that admit such drawings
  are RAC graphs.  RAC graphs are beyond-planar graphs and have been
  studied extensively.  In particular, it is known that a RAC graph
  with $n$ vertices has at most $4n-10$ edges.

  We introduce a superclass of RAC graphs, which we call
  \emph{arc-RAC} graphs.  A graph is arc-RAC
  if it admits a drawing where edges are represented by
  circular arcs and crossings form right angles.  We
  provide a Tur\'an-type result showing that an arc-RAC graph with $n$
  vertices has at most $14n-12$ edges and that there are $n$-vertex
  arc-RAC graphs with $4.5n - O(\sqrt{n})$ edges.
\end{abstract}

\section{Introduction}

A \emph{drawing} of a graph in the plane is a mapping of its vertices
to distinct points and each edge $uv$ to a curve whose endpoints are
$u$ and $v$.  Planar graphs, which admit crossing-free drawings, have
been studied extensively.  They have many nice properties and several
algorithms for drawing them are known, see,
e.g.,~\cite{juenger-book,kaufmann-book}.  However, in practice we must
also draw non-planar graphs and crossings make it difficult to
understand a drawing.
For this reason, graph classes with restrictions on crossings are
studied, e.g., graphs that can be drawn with at most $k$ crossings per edge 
(known as \emph{$k$-planar graphs}) or where the angles formed by each crossing are ``large''.
These classes are categorized as \emph{beyond-planar} graphs and have
experienced increasing interest in recent
years~\cite{DBLP:journals/csur/DidimoLM19}.

As introduced by Didimo et al.~\cite{del-dgrac-TCS11}, a prominent
beyond-planar graph class that concerns the crossing angles is the
class of $k$-bend right-angle-crossing graphs, or \emph{\rac{k}}
graphs for short, that admit a drawing where all crossings form
$90^\circ$ angles and each edge is a polygonal chain with at most~$k$
bends.
Using right-angle crossings and few bends is motivated by several cognitive studies
suggesting a positive correlation between large crossing angles or small curve
complexity and the readability of a graph drawing~\cite{Huang2007,Huang2014,Huang2008}. 
Didimo et al.~\cite{del-dgrac-TCS11} studied the edge density of
\rac{k} graphs.  They showed that \rac{0} graphs with $n$ vertices
have at most $4n-10$ edges (which is tight), that \rac{1} graphs have
at most $O(n^{\frac{4}{3}})$ edges, that \rac{2} graphs have at most
$O(n^{\frac{7}{4}})$ edges and that all graphs are \rac{3}.
Dujmovi\'c et al.~\cite{dgmw-nlacg-2010} 
gave an alternative simple proof of the $4n-10$ bound 
for \rac{0} graphs using charging arguments
similar to those of Ackerman and
Tardos~\cite{ackerman-tardos} and Ackerman~\cite{ackerman}.
Arikushi et al.~\cite{afkmt-garacd-CGTA12} improved the upper bounds
to $6.5n-13$ for \rac{1} graphs and to $74.2n$ for \rac{2} graphs. 
The bound of $6.5n-13$ for \rac{1} graphs was also obtained by charging arguments. 
They also provided
a \rac{1} graph with $4.5n - O(\sqrt{n})$ edges.
The best known lower and upper bound for the maximum edge density of
\rac{1} graphs of $5n-10$ and $5.5n-11$, respectively, are due to
Angelini et al.~\cite{abfk-rdgobe-2018}.

We extend the class of \rac{0} graphs by allowing edges to be drawn as circular 
arcs but still requiring $90^\circ$ crossings. 
An angle at which two circles intersect is the angle between the two
tangents to each of the circles at an intersection point. Two circles
intersecting at a right angle are called \emph{orthogonal}.
For any circle~$\gamma$, let $\ce{\gamma}$ be its center and let
$\ra{\gamma}$ be its radius. The following observation follows from the Pythagorean theorem.

\begin{figure}[tb]
  \centering
  \includegraphics{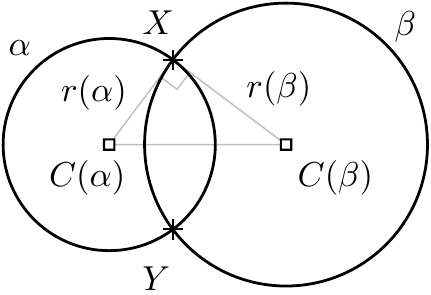}
  \caption{Circles~$\alpha$ and $\beta$ are orthogonal if and only if
    $\triangle X\ce{\alpha}\ce{\beta}$ is right-angled.}
  \label{fig:digon}
\end{figure}
\begin{observation}
  \label{obs:basic}
  Let $\alpha$ and $\beta$ be two circles.
  Then $\alpha$ and $\beta$ are orthogonal if and only if
  $\ra{\alpha}^2 + \ra{\beta}^2 = |\ce{\alpha}\ce{\beta}|^2$;
  see Figure~\ref{fig:digon}.
\end{observation}
In addition we note the following.
\begin{observation}
  \label{obs:tangent-center}
  Given a pair of orthogonal circles, the tangent to one circle at one of the
  intersection points goes through the center of the other circle; see
  Figure~\ref{fig:digon}.  In particular, a line is orthogonal to a
  circle if the line goes through the center of the circle.
\end{observation}

Similarly, two circular arcs $\alpha$ and $\beta$ are orthogonal if they intersect properly (that is, ignoring intersections at endpoints)
and the underlying circles (that contain the
arcs) are orthogonal.  For the remainder of this paper, all arcs will
be circular arcs.  We consider any straight-line segment to be an arc with 
infinite radius.  Note, though, that the above observations
do not hold for (pairs of) circles of infinite radius.  As in the case of
circles, for any arc~$\gamma$ of finite radius, let $\ce{\gamma}$ be
its center.

We call a drawing of a graph  an \emph{\arac\ drawing} if the edges are
drawn as arcs and any pair of intersecting arcs is orthogonal; see
Figure~\ref{fig:example}.  A graph that admits an
\arac drawings is called an \emph{\arac graph}.
\begin{figure}[tb]
  \centering
  \includegraphics{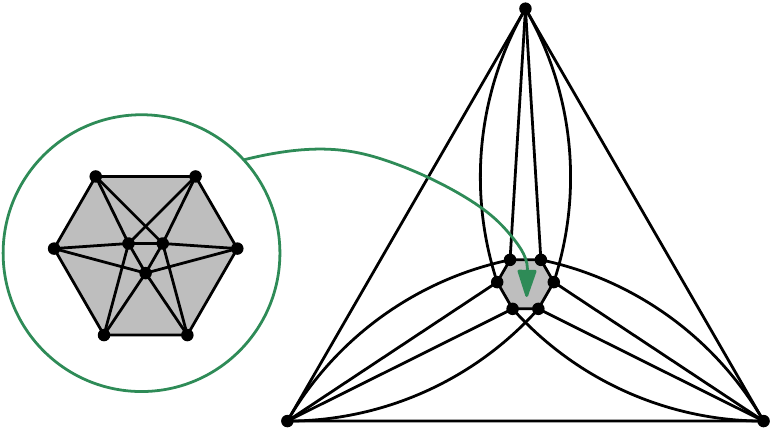}
  \caption{An \arac drawing of a graph. This graph is not \rac{0} \cite{bbhnr-nicpg-DAM17}.}
  \label{fig:example}
\end{figure}

\mktodo{remove the paragraph below when merging into the thesis}
The idea of drawing graphs with arcs dates back to at least the work
of the artist Mark Lombardi who drew social networks, featuring
players from the political and financial sector
\cite{lh-mkgn-IC-2003}.  Indeed, user studies
\cite{phnk-otuflgd-gd-2013,xrph-ausoceigv-DRAI2-12} state that users
prefer edges drawn with curves of small curvature; not necessarily for
performance (that is, tasks such as finding shortest paths,
identifying common neighbors, or determining vertex degrees) but for
aesthetics.  Drawing graphs with arcs can help to
improve certain quality measures of a drawing such as angular
resolution \cite{cdgk-dpgwca-DCG01,aacjr-twca-gd-2012}
or visual complexity \cite{s-dgfa-JGAA15,krw-dgfa-JGAA19}.

An immediate restriction on the edge density of \arac graphs is
imposed by the following known result.

\begin{lemma}[\cite{excursions}]
  \label{lem:no4arcs}
  In an \arac drawing, there cannot be four pairwise orthogonal arcs.
\end{lemma}

It follows from Lemma~\ref{lem:no4arcs} that \arac graphs are
\emph{\fqp}, that is, an \arac drawing cannot have four edges that
pairwise cross.  This implies that an \arac graph with $n$ vertices
can have at most $72(n - 2)$ edges~\cite{ackerman}.

Our main contribution is that we reduce this bound to $14n-12$ using
charging arguments similar to those of Ackerman~\cite{ackerman} and
Dujmovi\'c et al.~\cite{dgmw-nlacg-2010}; see
Section~\ref{sec:upper-bound}.  For us, the main challenge was to
apply these charging arguments to a modification
of an \arac drawing and to exploit, at the same time, geometric
properties of the original \arac drawing to derive the bound.  We also
provide a lower bound of $4.5n - O(\sqrt{n})$ on the maximum edge
density of \arac graphs based on the construction of Arikushi et
al.~\cite{afkmt-garacd-CGTA12}; see Section~\ref{sec:lower-bound}.  We
conclude with some open problems in Section~\ref{sec:open-problems}.
Throughout the paper our notation won't distinguish between the
entities (vertices and edges) of an abstract graph and the geometric
objects (points and curves) representing them in a drawing.

As usual for topological drawings, we forbid vertices to lie in the
relative interior of an edge and we do not allow edges to
\emph{touch}, that is, to have a common point in their relative
interiors without crossing each other at this point.  Hence an
\emph{intersection point} of two edges is always a \emph{crossing}.
When we say that two edges \emph{share a point},
we mean that they either cross each other or 
have a common endpoint.

\section{An Upper Bound for the Maximum Edge Density}
\label{sec:upper-bound}

Let $G$ be a \fqp graph, and let $D$ be a \fqp drawing of~$G$.  In his
proof of the upper bound on the edge density of \fqp graphs,
Ackerman~\cite{ackerman} first modified the given drawing so as to remove
faces of small degree.  We use a similar modification that we now
describe.

Consider two edges $e_1$ and $e_2$ in $D$ that intersect 
multiple times.
A region in $D$ bounded by pieces of $e_1$ and $e_2$ that
connect two consecutive crossings
or a crossing and a vertex of $G$
is called a \emph{lens}. 
If a lens is adjacent to a crossing and a vertex of $G$, then we call
such a lens a \emph{1-lens}, otherwise a \emph{0-lens}. 
A lens that does not contain a vertex of~$G$ is \emph{empty}.
\begin{figure}[t]
    \centering
    \includegraphics[page=2]{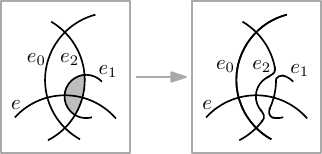}
    \caption{A simplification step resolves a smallest empty 0-lens; 
	     if two edges $e_1$ and $e_2$ change the order in which they cross the edge $e$, they form an empty 0-lens intersecting $e$ before the step, and thus, in the original \fqp drawing.}
    \label{fig:simplification-step}
\end{figure}
Every drawing with 0-lenses
has a \emph{smallest}
empty 0-lens, that is, an empty 0-lens that does not contain any other
empty 0-lenses in its interior.  
We can swap~\cite{prt-rptg-2006, ackerman} the
two curves that bound a smallest empty 0-lens; see
Figure~\ref{fig:simplification-step}.
We call such a swap a \emph{simplification step}.
Since a simplification step resolves a smallest
empty lens,
we observe the following.
\begin{observation}
  \label{obs:no-new-crossings-or-lenses}
  A simplification step does not introduce any new pairs of crossing
  edges or any new empty lenses.
\end{observation}
We exhaustively apply simplification steps to our drawing and refer to this as the \emph{simplification process}. 
Observation~\ref{obs:no-new-crossings-or-lenses} guarantees that 
applying the \emph{simplification process} to a drawing~$D$ terminates, that is, it results in 
an empty-0-lens-free drawing~$D'$ of~$G$.  
We call the resulting drawing $D'$ 
\emph{simplified}; it is a \emph{simplification} of~$D$. 
Observation~\ref{obs:no-new-crossings-or-lenses} implies the following
important property of any simplification step.

\begin{observation}
  \label{obs:remains-fqp}
  Applying a simplification step to a \fqp\ drawing yields a \fqp drawing. 
\end{observation}
As mentioned above, Ackerman~\cite{ackerman} used a similar
modification to prepare a \fqp drawing for his charging arguments;
note, that unlike Ackerman, we do not resolve 1-lenses.
We look at the simplification process in more detail, in particular, 
we consider how it changes the order in which 
edges cross. 
\begin{lemma}
  \label{lem:change-order}
  Let $D$ be an \arac drawing, and let $D'$ be a simplification
  of~$D$.  If two edges $e_1$ and $e_2$ cross another edge~$e$ in~$D'$
  in an order different from that in~$D$, then $e_1$ and $e_2$ form an
  empty 0-lens intersecting~$e$ in~$D$.
\end{lemma}
\begin{proof}
  Let $e_1$ and $e_2$ be two edges as in the statement of the lemma.
  Then there is a simplification step~$i$ where the order in which
  $e_1$ and~$e_2$ cross~$e$ changes.  Let~$D_i$ be the drawing
  immediately before simplification step~$i$, and let~$D_{i+1}$ be the
  drawing right after step~$i$.  By construction, the order in which
  $e_1$ and $e_2$ cross $e$ is different in~$D_i$ and in~$D_{i+1}$.
  Since $D_i$ is \fqp (see Observation~\ref{obs:remains-fqp}) and
  since we always resolve a smallest empty 0-lens, the edges $e_1$ and
  $e_2$ form a smallest empty 0-lens in $D_i$; see
  Figure~\ref{fig:simplification-step}.  Given that the simplification
  process does not introduce new empty lenses (see
  Observation~\ref{obs:no-new-crossings-or-lenses}), $e_1$ and~$e_2$
  form an empty 0-lens in the original \fqp drawing.
\end{proof}

We now focus on the special type of \fqp drawings we are interested
in.  Suppose that $G$ is an \arac graph, $D$ is an \arac drawing
of~$G$, and $D'$ is a simplification of~$D$.
Note that, in general, $D'$ is not an \arac drawing.
If two edges~$e_1$ and~$e_2$ cross in~$D'$,
then they do not form an empty 0-lens in $D$. 
This holds because for any two edges forming an empty 0-lens in $D$, the
simplification process removes both of their crossings; therefore,
in $D'$ the two edges do not have any crossings.
If $e_1$ and $e_2$ are incident to the same vertex, they also do not
form an empty 0-lens in~$D$,
as otherwise they would share three points in $D$ (the two crossing points of the lens and the common vertex of $G$).
Thus, we have the following observation.
\begin{observation}
  \label{obs:share-point-no-lens}
  Let $D$ be an \arac drawing, and let $D'$ be a simplification of~$D$.
  If two edges $e_1$ and $e_2$ share a point in~$D'$, then they do not
  form an empty 0-lens in~$D$.
\end{observation}

In the following, we first state the main theorem of this section and provide the structure of its proof (deferring one small lemma and the main technical lemma until later). 
Then, we prove the remaining technical details in Lemmas~\ref{lem:1-lenses} to~\ref{lem:0-pentagon-at-most-4-triangles-main} to establish the result.

\begin{theorem}
  \label{thm:density_upper_bound}
  An \arac graph with $n$ vertices can have at most $14n-12$ edges.
\end{theorem}
\begin{proof}
  Let $G=(V, E)$ be an \arac graph, let~$D$ be an \arac drawing of~$G$,
  let $D'$ be a simplification of~$D$, and let $G' = (V', E')$ be
  the planarization of $D'$.  Our charging argument consists of three
  steps.

  First, each face~$f$ of $G'$ is assigned an initial charge
  $\ch{f} = |f| + v(f) - 4$, where $|f|$ is the degree of~$f$ in the
  planarization and $v(f)$ is the number of vertices of~$G$ on the
  boundary of~$f$.  Applying Euler's formula several times, Ackerman
  and Tardos~\cite{ackerman-tardos} showed that
  $\sum_{f\in G'} \ch{f} = 4n-8$, where $n$ is the number of vertices
  of~$G$.  In addition, we set the charge $\ch{v}$ of a vertex $v$ of
  $G$ to $16/3$.  Hence the total charge of the system is
  $4n-8 + 16n/3=28n/3-8$.

  In the next two steps (described below),
  similarly to Dujmovi\'c et al.~\cite{dgmw-nlacg-2010}, we
  redistribute the charges among faces of $G'$ and vertices of $G$ so
  that, for every face~$f$, the final charge $\chfin{f}$
  is at least $v(f)/3$ and the final charge of each vertex is
  non-negative.  Observing that
  \[
    28n/3-8 \ge \sum_{f\in G'} \chfin{f} \ge \sum_{f\in G'}
    v(f)/3 = \sum_{v\in G} \deg(v)/3 = 2|E|/3
  \]
  yields that the number of edges of $G$ is at most $14n-12$ as
  claimed.  (The second-last equality holds since both sides count the
  number of vertex--face incidences in~$G'$.)

After the first charging step above, it is easy to see that $\ch{f} \ge
v(f)/3$ holds if $|f| \ge 4$.
We call a face $f$ of $G'$ a $k$-triangle, $k$-quadrilateral, or $k$-pentagon
if $f$ has the corresponding shape and $v(f) = k$.
Similarly, we call a face of degree two a \emph{digon}.
Note that any digon is a 1-digon
since all empty 0-lenses have been simplified.

After the first charging step, each digon and each 0-triangle has a
charge of~$-1$, and each 1-triangle has a charge of~$0$.
Thus, in the second charging step, we need to find ${4}/{3}$ units of charge
for each digon, one unit of charge for each 0-triangle, and ${1}/{3}$
unit of charge for each 1-triangle. Note that all other faces including 
2- and 3-triangles already have sufficient charge.

To charge a digon $d$ incident to a vertex $v$ of~$G$, we decrease
$\ch{v}$ by $4/3$ and increase $\ch{d}$ by $4/3$; see
Figure~\ref{fig:charging-vertex}.  We say that~$v$ \emph{contributes}
charge to~$d$.

To charge triangles, we proceed similarly to Ackerman~\cite{ackerman}
and Dujmovi\'c et al.~\cite[Theorem~7]{dgmw-nlacg-2010}.

Consider a $1$-triangle $t_1$.  Let $v$ be the unique vertex incident
to $t_1$, and let $s_1\in E'$ be the edge of $t_1$ opposite of~$v$;
see Figure~\ref{fig:flank}.  Note that the endpoints of $s_1$ are
intersection points in $D'$.  Let $f_1$ be the face on the other side
of $s_1$.  If $f_1$ is a $0$-quadrilateral, then we consider its edge
$s_2 \in E'$ opposite to $s_1$ and the face $f_2$ on the other side of
$s_2$.  We continue iteratively until we meet a face $f_k$ that is not
a $0$-quadrilateral.
If $f_k$ is a triangle, then all the faces $t_1,f_1,f_2, \dots, f_k$
belong to the same empty 1-lens $l$  
incident to the vertex~$v$ of~$t_1$.  In this case, we
decrease $\ch{v}$ by $1/3$ and increase $\ch{t_1}$ by $1/3$; see
Figure~\ref{fig:charging-vertex}.  Otherwise, $f_k$ is not a triangle
and $|f_k| + v(f_k) - 4 \ge 1$ (see Figure~\ref{fig:flank}).  In this
case, we decrease $\ch{f_k}$ by $1/3$ and increase $\ch{t_1}$ by $1/3$.
We say that the face $f_k$ contributes charge to the triangle~$t_1$
\emph{over} its side $s_k$.

For a 0-triangle~$t_0$, we repeat the above charging over each side.
If the last face on our path is a triangle~$t'$, then $t_0$ and $t'$ are
contained in an empty 1-lens 
(recall that $D'$ does not contain
empty 0-lenses)
and $t'$ is a 1-triangle incident to 
a vertex~$v$ of $G$.  In
this case, we decrease $\ch{v}$ by $1/3$ and increase $\ch{t_0}$ by
$1/3$; see Figure~\ref{fig:vertex-charging-0-triangle}.

Thus, at the end of the second step, the charge of each digon and triangle $f$ is at least $v(f)/3$.  Note that the charge of~$f$ comes
either from a higher-degree face or from a vertex~$v$ incident to an
empty 1-lens containing~$f$.

In the third step, we do not modify the charging any more, but we need
to ensure that
\begin{enumerate}[(i)]
\item $\ch{f} \ge v(f)/3$ still holds for each face $f$ of $G'$ with
  $|f| \ge 4$ and
\item $\ch{v} \ge 0$ for each $v$ of $G$.
\end{enumerate}

We first show statement~(i).
Ackerman~\cite{ackerman} noted that a face $f$ with $|f| \ge 4$ can
contribute charges over each of its edges at most once.
Moreover, $f$ can contribute at most one third unit of charge over
each of its edges. Therefore, if $|f| + v(f) \ge 6$, then in the worst
case (that is, $f$ contributes charge over each of its edges) $f$
still has a charge of $|f| + v(f) - 4 - |f|/3 \ge v(f)/3$.
Thus, it remains to
verify that 1-quadrilaterals and 0-pentagons, which initially had only
one unit of charge, have a charge of at least $1/3$ unit or zero,
respectively, at the end of the second step.

A 1-quadrilateral~$q$ can contribute charge to at most two triangles
since the endpoints of any edge of $G'$ over which a face contributes
charge must be intersection points in $D'$; see
Figure~\ref{fig:charging-1-quadrilateral} and recall that~$q$ now
plays the role of $f_k$ in Figure~\ref{fig:flank}.

A 0-pentagon cannot contribute charge to more than three triangles;
see Lemma~\ref{lem:0-pentagon-at-most-4-triangles-main}.

Now we show statement~(ii).  Recall that a vertex $v$ can contribute
charge to a digon incident to~$v$ or to at most two triangles
contained in an empty 1-lens incident to~$v$.  Observe that two empty
1-lenses with either triangles or a digon taking charge from $v$
cannot overlap; see Figure~\ref{fig:charging-vertex}.  
We show in Lemma~\ref{lem:1-lenses} that
$v$ cannot be incident to more than four such empty 1-lenses.  In the
worst case, $v$ contributes ${4}/{3}$ units of charge to each of the
at most four incident digons representing these empty 1-lenses.  Thus,
$v$ has non-negative charge at the end of the second step. 
\end{proof}

\begin{figure}[tb]
    \begin{subfigure}[b]{0.24\textwidth}
      \centering
      \includegraphics{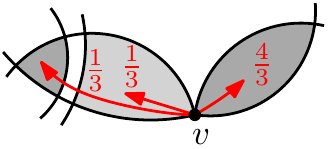}
      \caption{Vertex $v$ contributes charge to a digon and two
        triangles contained in empty 1-lenses.}
      \label{fig:charging-vertex}
     \end{subfigure}
     \hfill
     \begin{subfigure}[b]{0.18\textwidth}
      \centering
      \includegraphics[page=1]{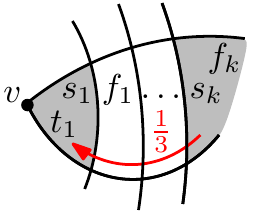}
       \caption{Face $f_k$ contributes charge to the 1-triangle $t_1$.}
       \label{fig:flank}
    \end{subfigure}
    \hfill
    \begin{subfigure}[b]{0.19\textwidth}
      \centering
      \includegraphics[page=2]{flank}
       \caption{Vertex $v$ contributes charge to the 0-triangle $t_0$.}
       \label{fig:vertex-charging-0-triangle}
    \end{subfigure}
    \hfill
    \begin{subfigure}[b]{0.22\textwidth}
      \centering
      \includegraphics{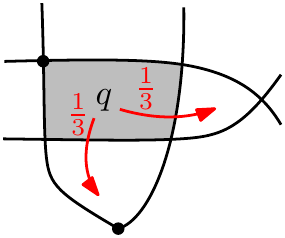}
      \caption{1-quadrilateral $q$ contributes charge to at most two
        triangles.}
       \label{fig:charging-1-quadrilateral}
     \end{subfigure}
     \caption{Transferring charge from vertices and high-degree faces
       to small-degree faces.}
     \label{fig:charging}
\end{figure}

\begin{lemma}
  \label{lem:1-lenses}
  In any simplified \arac drawing, each vertex is incident to at most
  four non-overlapping empty 1-lenses.
\end{lemma}

\begin{proof}
  Let $v$ be a vertex incident to some non-overlapping empty 1-lenses. 
  Consider a small neighborhood of the vertex $v$ in the simplified drawing and notice that in this neighborhood
  the simplified drawing is the same as the original \arac drawing. Let $l$ be one of
  the non-overlapping empty 1-lenses incident to $v$.  Then $l$ forms an angle of $90^\circ$ between the
  two edges incident to~$v$ that form $l$; see
  Figure~\ref{fig:1-lenses}.  This is due to the fact that the other
  ``endpoint'' of~$l$ is an intersection point where the two
  edges must meet at $90^\circ$.  Thus $v$ is incident to at most four non-overlapping
  empty 1-lenses.
\end{proof}

\begin{figure}[tb]
  \begin{minipage}[b]{.35\textwidth}
    \centering
    \includegraphics{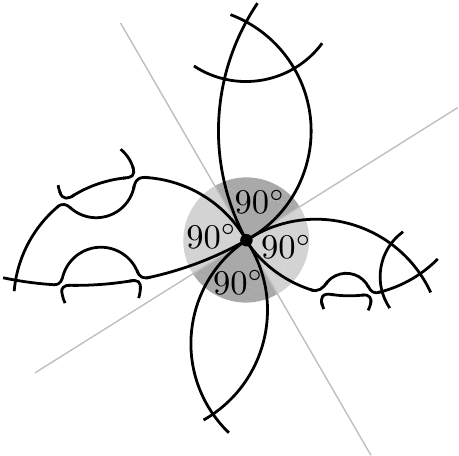}
    \caption{The edges of an empty 1-lens
      form a $\pi/2$ angle at the vertex of the lens.}
    \label{fig:1-lenses}
  \end{minipage}
  \hfill
  \begin{minipage}[b]{.6\textwidth}
    \centering
    \includegraphics{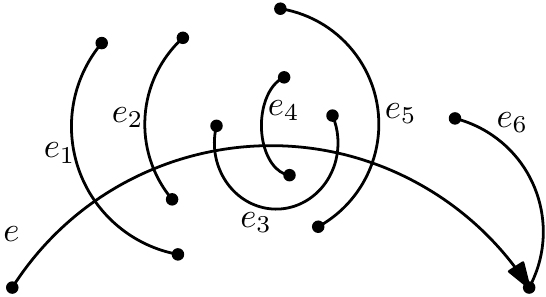}
    \caption{The relation $\ord{\cdot \;}{\,\cdot}$ does not
      necessarily describe \emph{all} intersection points along an
      edge.  Here, e.g., $\ord{e}{e_1, e_2, e_3, e_4,e_3, e_5, e_6}$
      and $\ord{e}{e_1, e_3, e_4, e_5}$ both hold.}
    \label{fig:order}
  \end{minipage}
\end{figure}

We now set the stage for proving
Lemma~\ref{lem:0-pentagon-at-most-4-triangles-main}, which shows that
a 0-pentagon in a simplified drawing does not contribute charge to 
more than three triangles.
The proof goes by a contradiction.
Consider a 0-pentagon that contributes
charge to at least four triangles in the simplified drawing. 
First, we examine which edges of this 0-pentagon cross; see Lemma~\ref{lem:edges-0-3}. 
We then describe the order in which these edges 
share points in the simplified drawing and show that
the original \arac drawing must adhere to the same order; see
Lemma~\ref{lem:edges-order}.
Finally, we use geometric arguments to show that,
under these order constraints,
an \arac drawing of the edges does not exist; see
Lemma~\ref{lem:0-pentagon-at-most-4-triangles}.

Let $D$ be an \arac drawing of some \arac graph $G = (V, E)$, let $D'$
be its simplification, and let $p$ be a 0-pentagon that contributes
charge to at least four triangles.  Let $s_0, s_1, \dots, s_4$ be the
sides of $p$ in clockwise order and denote the edges of $G$ that
contain these sides as~$e_0, e_1, \dots, e_4$ so that edge~$e_0$
contains side~$s_0$ etc.  Since $p$ contributes charge over at least four
sides, these sides are consecutive around $p$.  Without loss of
generality, we assume that $s_4$ is the side over which $p$ does not
necessarily contribute charge.

For $i \in \{ 0, 1, 2, 3\}$, let $t_i$ be the triangle that gets charge
from~$p$ over the side~$s_i$.
The triangle $t_i$ is bounded by the edges $e_{i-1}$ and $e_{i+1}$.
(Indices are taken modulo $5$.)
Note that all faces bounded by $e_{i-1}$ and $e_{i+1}$ that are between 
$t_i$ and $p$ must be 0-quadrilaterals. 
If $t_i$ is a 1-triangle, then $e_{i-1}$ and $e_{i+1}$ are incident to
the same vertex of the triangle.  Otherwise, $t_i$ is a 0-triangle and
$e_{i-1}$ and $e_{i+1}$ cross at a vertex of the triangle.
Let $A'_{i-1, i+1}$ denote this common point of $e_{i-1}$ and $e_{i+1}$, and let
$E_p=\{e_0, \dots, e_4\}$; see Figure~\ref{fig:0-pentagon-D'}.

\begin{figure}[tb]
    \begin{subfigure}[b]{0.48\textwidth}
      \centering
      \includegraphics[page=1]{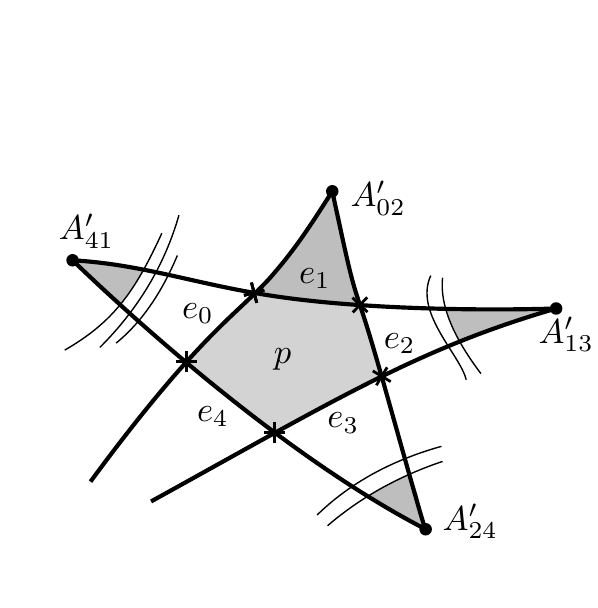}
      \caption{Notation: A 0-pentagon $p$ in $D'$ and the edges
        in~$E_p$.  The points of type $A'_{i-1,i+1}$ are either
        intersection points or vertices of $G$.}
      \label{fig:0-pentagon-D'}
    \end{subfigure}
    \hfill
    \begin{subfigure}[b]{0.44\textwidth}
      \centering
      \includegraphics[page=1]{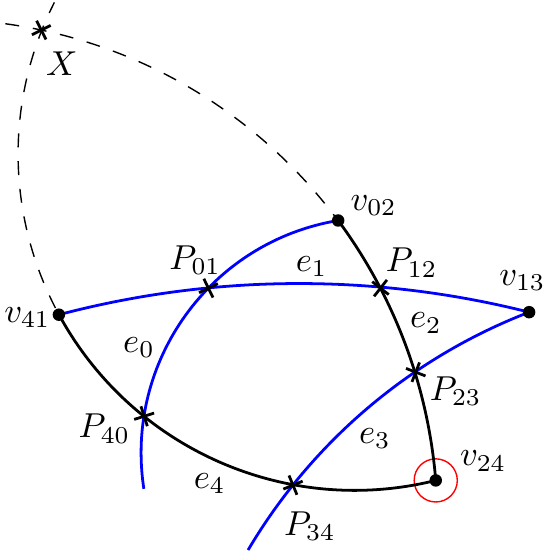}
      \caption{Also in $D$, it holds that
        $\ord{e_0}{e_{4}, e_{1}, e_{2}}$,
        $\ord{e_3}{e_{1}, e_{2}, e_{4}}$, and, for
        $i \in \{1, 2, 4\}$,
        $\ord{e_i}{e_{i-2}, e_{i-1}, e_{i+1}, e_{i+2}}$.}
      \label{fig:0-pentagon-D}
    \end{subfigure}
    \caption{A $0$-pentagon cannot contribute charge to more than three
      triangles.}
    \label{fig:0-pentagon}
\end{figure}

We now describe the order in which the edges in $E_p$ share points in $D'$.
To this end, we orient the edges in~$E_p$ so that this orientation conforms with
the orientation of a clockwise walk around the boundary of $p$ in $D'$. In addition, we
write $\ord{e_k}{e_{i_1}, e_{i_2}, \dots, e_{i_l}}$ if the edge~$e_k$ shares points
(either crossing points or vertices of the graph)
with the edges $e_{i_1}, e_{i_2}, \dots, e_{i_l}$ in this order with
respect to the orientation of~$e_k$; see Figure~\ref{fig:order}.
(Note that we can have $\ord{e_k}{e_i, e_j, e_i}$ as edges may
intersect twice. We will not consider more than two edges sharing the same endpoint.)  
Due to the order in which we numbered the edges in~$E_p$, it holds
in~$D'$ that $\ord{e_0}{e_{4}, e_{1}, e_{2}}$,
$\ord{e_3}{e_{1}, e_{2}, e_{4}}$, and, for $i \in \{1,2,4\}$,
$\ord{e_i}{e_{i-2}, e_{i-1}, e_{i+1}, e_{i+2}}$; see
Figure~\ref{fig:0-pentagon-D'}.
Now we show that in~$D$ the order is the same.
Obviously every pair of edges $(e_{i-1}, e_{i+1})$ that shares an
endpoint in~$D'$ also shares an endpoint in~$D$.
Furthermore, every pair $(e_i, e_{i+1})$ or
$(e_{i-1}, e_{i+1})$ of crossing edges crosses in~$D$, too,
because the simplification process does not introduce new pairs of
crossing edges; see Observation~\ref{obs:no-new-crossings-or-lenses}.

\begin{lemma}
  \label{lem:edges-0-3}
  In the drawing $D$, the edges $e_{0}$ and $e_{3}$ do not cross.
\end{lemma}

\begin{proof}
  Assume that the edges $e_{0}$ and $e_{3}$ cross in $D$ and notice that each of the
  pairs of edges $(e_{0}, e_{1})$, $(e_{1}, e_{2})$, and $(e_{2}, e_{3})$ 
  forms a crossing in $D'$ (see Figure~\ref{fig:0-pentagon-D'}), and
  hence in~$D$, too.
  For any arc $e$, let $\cv{e}$ denote the circle containing $e$.
  Recall that a family of \emph{Apollonian circles} \cite{excursions,
    cfkw-oaooc-GD19} consists of two sets of circles such that each
  circle in one set is orthogonal to each circle in the other set.
  Thus, the pairs of circles $(\cv{e}_{1}, \cv{e}_{3})$ and
  $(\cv{e}_{0}, \cv{e}_{2})$ belong to such a family; the pair
  $(\cv{e}_{1}, \cv{e}_{3})$ belongs to one set of the family and
  $(\cv{e}_{0}, \cv{e}_{2})$ belongs to the other set.
  If not all of the circles in the family share the same point, 
  which is the case for the circles $\cv{e}_{0}$, $\cv{e}_{1}$, $\cv{e}_{2}$, and $\cv{e}_{3}$, then one such set 
  consists of disjoint circles.
  So either the pair $(\cv{e}_{0}, \cv{e}_{2})$ or the pair
  $(\cv{e}_{1}, \cv{e}_{3})$ must consist of disjoint circles.
  This is a contradiction because each of the two pairs shares a point
  in~$D'$ (see Figure~\ref{fig:0-pentagon-D'}), and thus, in~$D$.
\end{proof}

\begin{lemma}
  \label{lem:edges-order}
  In the drawing $D$, it holds that $\ord{e_0}{e_{4},e_{1},e_{2}}$,
  $\ord{e_3}{e_{1},e_{2},e_{4}}$, and, for each $i \in \{1,2,4\}$,
  $\ord{e_i}{e_{i-2},e_{i-1},e_{i+1},e_{i+2}}$.
\end{lemma}  

\begin{proof}
  Recall that in the drawing~$D'$, it holds that $\ord{e_0}{e_{4},e_{1},e_{2}}$,
  $\ord{e_3}{e_{1},e_{2},e_{4}}$, and, for each $i \in \{1,2,4\}$,
  $\ord{e_i}{e_{i-2},e_{i-1},e_{i+1},e_{i+2}}$; see
  Figure~\ref{fig:0-pentagon-D'}.  
  Consider distinct indices $i,j,k \in\{0,1,2,3,4\}$ so that the edges
  $e_i$ and $e_j$ share points
  with $e_k$ in this order in~$D'$, that is, $\ord{e_k}{e_{i}, e_{j}}$ in $D'$.
  We will show that the edges $e_i$ and $e_j$ share points
  with $e_k$ in the same order in~$D$, that is, $\ord{e_k}{e_{i},
    e_{j}}$ in $D$.  In other words, the order in which the edges
  in~$E_p$ share points in~$D$ is the same as in~$D'$.
  
  First, note that if the edge~$e_i$ or the edge~$e_j$ shares an
  endpoint with~$e_k$, then $e_i$ and $e_j$ do not change the order in
  which they share points with~$e_k$.  This is due to the fact that
  the simplification process does not modify the
  graph.  Therefore, $e_i$ and $e_j$ share points with
  $e_k$ in the same order in~$D$ as in~$D'$, that is,
  $\ord{e_k}{e_{i}, e_{j}}$ in $D$.
  
  Assume now that both $e_i$ and $e_j$ cross $e_k$.
  
  If $(i,j) \in \{(0,3),(3,0)\}$, then, according to Lemma~\ref{lem:edges-0-3},
  the edges $e_i$ and $e_j$ do not cross
  in $D$, so they do not form an empty 0-lens in $D$, and thus, by Lemma~\ref{lem:change-order}, 
  $e_i$ and $e_j$ 
  cross
  $e_k$ in the same order in~$D$ as in~$D'$, that is, $\ord{e_k}{e_{i}, e_{j}}$ in $D$.

  Otherwise, the edges $e_i$ and $e_j$ share a point in $D'$; see Figure~\ref{fig:0-pentagon-D'}. 
  Therefore, by
  Observation~\ref{obs:share-point-no-lens}, $e_i$ and $e_j$ do not form
  an empty 0-lens in $D$, and thus, by
  Lemma~\ref{lem:change-order}, $e_i$ and $e_j$ 
  cross
  $e_k$ in the same order in~$D$ as in~$D'$, that is,
  $\ord{e_k}{e_{i}, e_{j}}$ in $D$.
\end{proof}

Thus, we have shown that the order in which
the edges in $E_p$ share points in~$D$ is the same as in~$D'$,
see Figure~\ref{fig:0-pentagon-D}.
We show now that an \arac drawing with this order
does not exist; see Lemma~\ref{lem:0-pentagon-at-most-4-triangles}.
This is the main ingredient to prove
Lemma~\ref{lem:0-pentagon-at-most-4-triangles-main}, which says that a
0-pentagon in a simplified arc-RAC drawing contributes charge to at
most three triangles.

For simplicity of presentation and without loss of generality, we assume
that the points $A'_{i-1, i+1}$ are vertices of $G$, which we denote
by $v_{i-1, i+1}$. 

\begin{lemma}
  \label{lem:0-pentagon-at-most-4-triangles}
  The edges in $E_p$ do not admit an \arac drawing where it holds that
  $\ord{e_0}{e_{4}, e_{1}, e_{2}}$, $\ord{e_3}{e_{1}, e_{2}, e_{4}}$,
  and, for $i \in \{1,2,4\}$,
  $\ord{e_i}{e_{i-2}, e_{i-1}, e_{i+1}, e_{i+2}}$.
\end{lemma}

\begin{proof}
  Assume that the edges in~$E_p$ admit an \arac drawing where they
  share points in the order indicated above.  For
  $i \in \{0,\dots,4\}$, let $P_{i, i+1}$ be the intersection point of
  $e_{i}$ and $e_{i+1}$; see Figure~\ref{fig:0-pentagon-D}.  Note that
  on~$e_i$, the point $P_{i-1, i}$ is before the point $P_{i, i+1}$
  (due to $\ord{e_i}{e_{i-1}, e_{i+1}}$).

  \mktodo{remove this paragraph}
  Recall that an \emph{inversion}~\cite{excursions} with respect to a
  circle $\alpha$, the \emph{inversion circle}, is a mapping that takes
  any point $P \neq \ce{\alpha}$ to a point $P'$ on the straight-line ray from
  $\ce{\alpha}$ through~$P$ so that $|\ce{\alpha}P'|\cdot|\ce{\alpha}P| =
  \ra{\alpha}^2$.  Inversion maps each circle not passing through
  $\ce{\alpha}$ to another circle and each circle passing through
  $\ce{\alpha}$ to a line.  The center of the inversion circle is
  mapped to the ``point at infinity''.  It is known that inversion
  preserves angles.
  
  We invert the drawing of the edges in~$E_p$ with respect to a small
  inversion circle centered at~$v_{24}$.
  Let $\img{e}_{i}$ be the image of $e_{i}$, $\img{v}_{i-1,i+1}$ be the image of $v_{i-1,i+1}$ 
  ($\img{v}_{24}$ is the point at infinity), and
  $\img{P}_{i,i+1}$ be the image of $P_{i,i+1}$.
  Because in the pre-image the arcs $e_{2}$ and $e_{4}$ pass through $v_{24}$, 
  in the image $\img{e}_{2}$ and $\img{e}_{4}$ are straight-line rays.
  We assume that in the image $\img{e}_{2}$ meets 
  $\img{e}_{4}$ at the point at infinity, that is, at $\img{v}_{24}$. Then, taking into account
  that inversion is a continuous and injective mapping, the order in which the edges in $E_p$ 
  share points is the same in the image.

  We consider two cases regarding whether the edges $e_{2}$ and $e_{4}$ belong to two
  different circles or not.
  
  \medskip\noindent\emph{Case~I:}
  $e_{2}$ and $e_{4}$ belong to two different circles.
  
  One of the intersection points of their circles is $v_{24}$, 
  and we let $X$ denote the other intersection point.
  Here we have that $\img{e}_{2}$ and $\img{e}_{4}$ are two
  straight-line rays meeting at infinity at~$\img{v}_{24}$.
  Their supporting lines are different and intersect 
  at~$\img{X}$, which is the image of $X$; see
  Figure~\ref{fig:0-pentagon-inversion-different}.

  \begin{figure}[tb]
    \begin{subfigure}[b]{0.48\textwidth}
      \centering
      \includegraphics{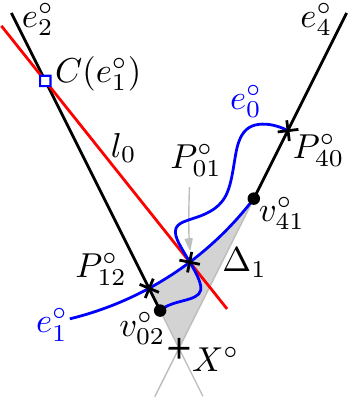}
      \caption{}
      \label{fig:0-new-proof}
    \end{subfigure}
    \hfill
    \begin{subfigure}[b]{0.48\textwidth}
      \centering
      \includegraphics{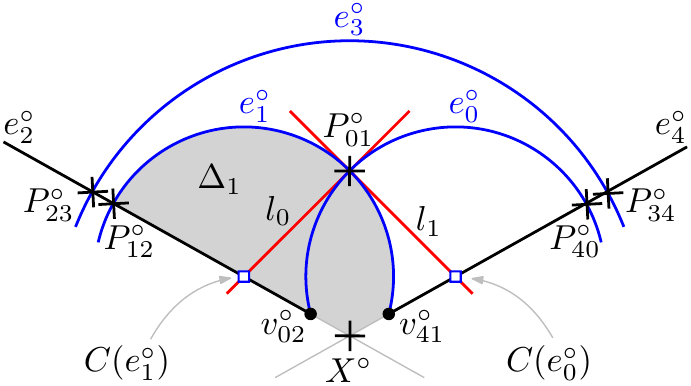}
      \caption{}
      \label{fig:0-new-proof-old-part}
    \end{subfigure}
    \caption{Illustration for the proof of
      Lemma~\ref{lem:0-pentagon-at-most-4-triangles} when
      $e_{2}$ and $e_{4}$ belong to two different circles.
      Image of the inversion with respect to the red circle in
      Figure~\ref{fig:0-pentagon-D}.}
    \label{fig:0-pentagon-inversion-different}
  \end{figure}

  We now assume for a contradiction that the arc $\img{e}_{1}$ forms
  a concave side of the triangle $\Delta_1=\img{P}_{12}\img{v}_{41}\img{X}$;
  see Figure~\ref{fig:0-new-proof} where the triangle is filled gray.
  (Symmetrically, we can show that the
  arc $\img{e}_{0}$ cannot form a concave side of the triangle
  $\Delta_0=\img{P}_{40}\img{v}_{02}\img{X}$.)
  By Observation~\ref{obs:tangent-center}, $\ce{\img{e}_{1}}$ must lie
  on the ray~$\img{e}_2$.  Since we assume that the arc $\img{e}_{1}$
  forms a concave side of the triangle $\Delta_1$, $\ce{\img{e}_{1}}$ and
  $\img{v}_{02}$ are separated by~$\img{P}_{12}$ on $\img{e}_2$.
  Consider the tangent~$l_0$ to $\img{e}_{0}$ at $\img{P}_{01}$.  Again
  in light of Observation~\ref{obs:tangent-center}, $l_0$ has to go
  through $\ce{\img{e}_{1}}$ because $\img{e}_{0}$ and $\img{e}_{1}$
  are orthogonal.
  On the one hand, $\img{v}_{02}$ is to the same side of $l_0$ as
  $\img{P}_{12}$; see Figure~\ref{fig:0-new-proof}.  On the other
  hand, $l_0$ separates $\img{P}_{12}$ and $\img{v}_{41}$ due to
  $\ord{e_1}{e_4,e_0,e_2}$.  Moreover, $l_0$ does not
  separate~$\img{v}_{41}$ and~$\img{P}_{40}$ since it intersects the
  line of~$\img{e}_4$ when leaving the gray triangle~$\Delta_1$.  So
  the two points $\img{v}_{02}$ and $\img{P}_{40}$ of the same arc
  $\img{e}_{0}$ are separated by~$l_0$, which is a tangent of this
  arc; contradiction.
  
  Thus, the arc $\img{e}_{1}$ forms a convex side of the triangle
  $\Delta_1$, and $\img{e}_{0}$ forms a convex side of~$\Delta_0$;
  see Figure~\ref{fig:0-new-proof-old-part}.
  Now, due to Observation~\ref{obs:tangent-center}, 
  $\ce{\img{e}_{0}}$ is between $\img{v}_{41}$ and $\img{P}_{40}$, and
  $\ce{\img{e}_{1}}$ is between $\img{v}_{02}$ and $\img{P}_{12}$,
  because that is where the tangents $l_1$ of $\img{e}_1$ and $l_0$ of
  $\img{e}_0$ in $\img{P}_{01}$
  intersect the lines of $\img{e}_{4}$ and $\img{e}_{2}$, respectively.
  Taking into account that $\ce{\img{e}_{3}} = \img{X}$, because
  $\img{e}_{3}$ is orthogonal to both $\img{e}_{2}$ and $\img{e}_{4}$, we obtain
  that the points
  $\ce{\img{e}_{3}}$, $\ce{\img{e}_{1}}$, $\img{P}_{12}$, $\img{P}_{23}$ 
  appear on the line of $\img{e}_{2}$ in this order.
  Thus, the circle of $\img{e}_{1}$ is contained within the circle of
  $\img{e}_{3}$.  This is a contradiction because $\img{e}_{3}$ and
  $\img{e}_{1}$ must share a point; namely $\img{v}_{13}$.
    
  \medskip\noindent\emph{Case~II:} $e_{2}$ and $e_{4}$ belong to the
  same circle.

  Here
  $\img{e}_{2}$ and $\img{e}_{4}$ are two disjoint straight-line rays
  on the same line~$l$ (meeting at infinity at $\img{v}_{24}$); 
  see Figure~\ref{fig:0-pentagon-inversion-same}.  We
  direct~$l$ as~$\img{e}_4$ and~$\img{e}_2$ (from right to left in
  Figure~\ref{fig:0-pentagon-inversion-same}).  Because $\img{e}_{0}$,
  $\img{e}_{1}$, and $\img{e}_{3}$ are orthogonal to $l$, their
  centers have to be on $l$.  Due to our initial assumption, we have
  $\ord{e_4}{e_2,e_3,e_0,e_1}$ and $\ord{e_2}{e_0,e_1,e_3,e_4}$.  Hence,
  along~$l$, we have $\img{P}_{34}$, $\img{P}_{40}$, $\img{v}_{41}$,
  (on~$\img{e}_4$) and then $\img{v}_{02}$, $\img{P}_{12}$,
  $\img{P}_{23}$ (on~$\img{e}_2$).
  Therefore, the circle of~$\img{e}_1$ is contained in that
  of~$\img{e}_3$.  Hence, $\img{e}_{1}$ does not share a point with
  $\img{e}_{3}$; a contradiction.
\end{proof}

\begin{figure}[tb]
  \centering
  \includegraphics{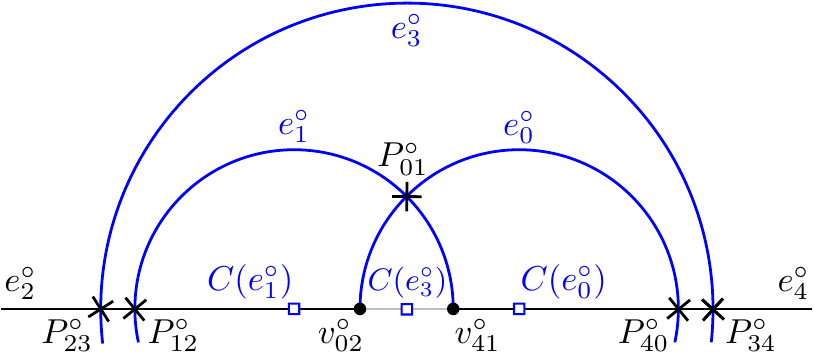}
  \caption{Illustration to the proof of
    Lemma~\ref{lem:0-pentagon-at-most-4-triangles} when $e_{2}$ and
    $e_{4}$ belong to the same circle.  Image of the inversion with
    respect to the red circle in Figure~\ref{fig:0-pentagon-D}.}
  \label{fig:0-pentagon-inversion-same}
\end{figure}

\begin{lemma}
  \label{lem:0-pentagon-at-most-4-triangles-main}
  A 0-pentagon in a simplified \arac drawing contributes charge to at
  most three triangles.
\end{lemma}

\begin{proof}
  As discussed above, if a 0-pentagon formed by edges $e_0, e_1,
  \dots, e_4$ contributes charge to more than three triangles in a
  simplified drawing (see Figure~\ref{fig:0-pentagon-D'}), then
  this implies the existence 
  of an \arac drawing where it holds that
  $\ord{e_0}{e_{4}, e_{1}, e_{2}}$, $\ord{e_3}{e_{1}, e_{2}, e_{4}}$
  and, for $i \in \{1, 2, 4\}$,
  $\ord{e_i}{e_{i-2}, e_{i-1}, e_{i+1}, e_{i+2}}$;
  see Figure~\ref{fig:0-pentagon-D}.  This, however, contradicts
  Lemma~\ref{lem:0-pentagon-at-most-4-triangles}.
\end{proof}

With the proofs of Lemmas~\ref{lem:1-lenses}~and~\ref{lem:0-pentagon-at-most-4-triangles-main} now in place, the proof of Theorem~\ref{thm:density_upper_bound} is complete.

\section{A Lower Bound for the Maximum Edge Density}
\label{sec:lower-bound}

In this section, we construct a family of \arac graphs with high edge
density.  Our construction is based on a family of \rac{1} graphs of
high edge density that Arikushi et al.~\cite{afkmt-garacd-CGTA12} constructed.
Let $G$ be an embedded graph whose vertices are the vertices of the
hexagonal lattice clipped inside a rectangle; see Figure~\ref{fig:tiling}.
The edges of $G$ are the edges of the lattice and, inside each hexagon
that is bounded by the cycle $(P_0,\ldots,P_5)$, six additional edges
$(P_i, P_{i+2 \bmod 6})$ for $i \in \{0,1,\dots,5\}$; see Figure~\ref{fig:tile}.
We refer to a part of the drawing made up of a single hexagon and its
diagonals as a \emph{tile}.  In Theorem~\ref{thm:density_lower_bound}
below, we show that each hexagon can be drawn as a regular hexagon and its
diagonals can be drawn as two sets of arcs $A = \{\alpha_0, \alpha_1,
\alpha_2\}$ and $B= \{\beta_0, \beta_1, \beta_2\}$, so that the arcs in $A$ are
pairwise orthogonal, the arcs in $B$ are pairwise non-crossing, and
for each arc in $B$ intersecting another arc in $A$ the two arcs are
orthogonal; we use this construction to establish the theorem.
In particular, the arcs in $A$ form the 3-cycle $(P_0,P_2,P_4)$,
and the arcs in $B$ form the 3-cycle $(P_1,P_3,P_5)$.

We first define the radii and centers of the
arcs in a tile and show that they form only orthogonal crossings. 
We use the geometric center of the tile as the origin of our
coordinate system in the following analysis.
We now discuss the arcs in~$A$; then we turn to the arcs in~$B$.
For each $j \in \{0, 1, 2\}$, the arc $\alpha_j$ has radius $r_A = 1$
and center $\ce{\alpha_j} = (d_A\cos(\pi/6 + j \frac{2\pi}{3}),
d_A\sin(\pi/6 + j \frac{2\pi}{3}))$, where $d_A=\sqrt{2/3}$ is the
distance of the centers from the origin;
see Figure~\ref{fig:confA}.

\begin{lemma}
 \label{lem:3circles}
  The arcs in $A$ are pairwise orthogonal.
\end{lemma}

\begin{proof}
  Consider the equilateral triangle $\triangle
  \ce{\alpha_0}\ce{\alpha_1}\ce{\alpha_2}$ formed by the centers of the three
  arcs in~$A$.  Because the origin is in the center of the triangle,
  the edge length of the triangle is $2 d_A\cos{\pi/6} =
  \sqrt{2}$, and so the distance between the centers of any two
  arcs is $\sqrt{2}$.  The radii of the arcs are~1, hence
  by Observation~\ref{obs:basic},
  every two arcs are orthogonal.
\end{proof}

As in Figure~\ref{fig:confB}, for each $j \in \{0,1,2\}$, the arc
$\beta_j$ has radius $r_B = \sqrt{\frac{70 + 40\sqrt{3}}{6}}$ 
and center $\ce{\beta_j} = (d_B\cos(\frac{\pi}{2} + 
\frac{(j+1)2\pi}{3}), d_B\sin(\frac{\pi}{2} +  \frac{(j+1)2\pi}{3}))$,
where $d_B = \sqrt\frac{1}{6} + \sqrt{\frac{73 + 40\sqrt{3}}{6}}$ is
the distance of the centers from the origin.

\begin{lemma}
  \label{lem:3big_circles}
  If an arc in $B$ intersects an
  arc in $A$, then the two arcs are orthogonal.
\end{lemma}

\begin{proof}
  Let $i, j \in \{0,1,2\}$. If $j \ne i$, $\|\ce{\alpha_i}
  -\ce{\beta_j}\|^2 =  \frac{76 + 40\sqrt{3}}{6} = 1 + \frac{70 +
    40\sqrt{3}}{6} = r_A^2 + r_B^2$, so
  by Observation~\ref{obs:basic} $\alpha_i$ and $\beta_j$ are 
  orthogonal.  Otherwise, for $i \in \{0,1,2\}$, $\|\ce{\alpha_i}
  -\ce{\beta_i}\| = \sqrt{\frac{112+64\sqrt{3}}{6}} > 1 +
  \sqrt{\frac{70 + 40\sqrt{3}}{6}} = r_A + r_B$,
  so $\alpha_i$ and $\beta_i$ do not intersect.
\end{proof}

\begin{figure}[tb]
    \begin{subfigure}[b]{0.48\textwidth}
    \centering
    \includegraphics{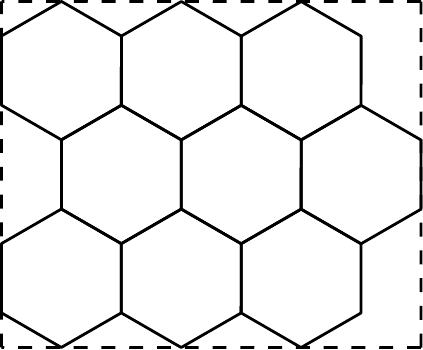}
    \caption{The hexagonal lattice.}
    \label{fig:tiling}
    \end{subfigure}
    \hfill
    \begin{subfigure}[b]{0.48\textwidth}
      \centering
      \includegraphics{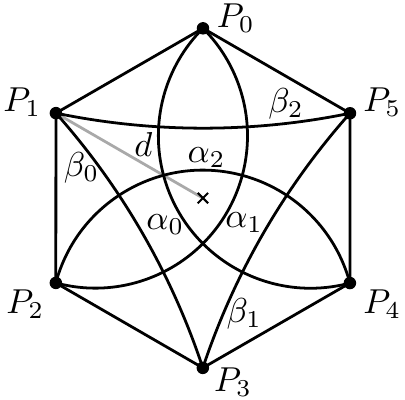}
      \caption{A tile.}
      \label{fig:tile}
    \end{subfigure}
    \caption{Tiling used for the lower-bound construction.}
\end{figure}

\begin{figure}[tb]
    \begin{subfigure}[b]{0.48\textwidth}
      \centering
      \includegraphics{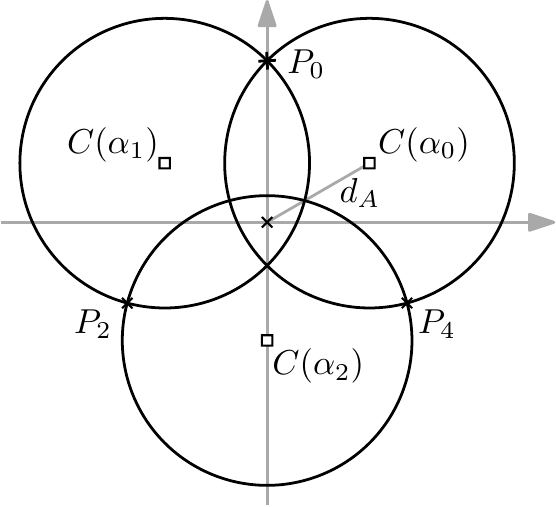}
      \caption{Circles covering the arcs in~$A$.}
      \label{fig:confA}
    \end{subfigure}
    \hfill
    \begin{subfigure}[b]{0.48\textwidth}
      \centering
      \includegraphics{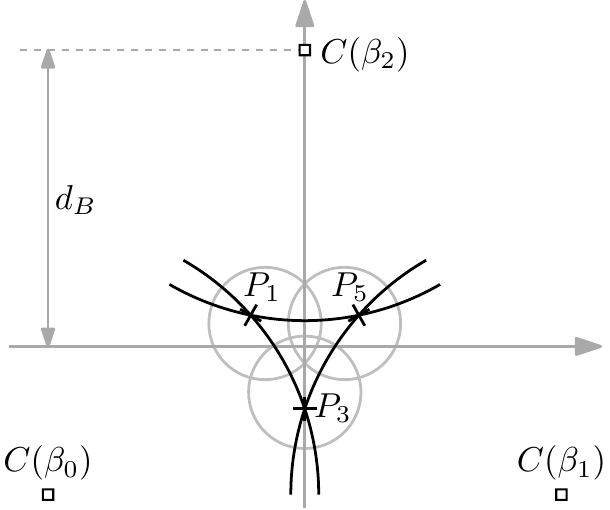}
      \caption{Circles covering the arcs in~$B$.}
      \label{fig:confB}
     \end{subfigure}
    \caption{Construction for the lower bound on the maximum edge density of
      arc-RAC graphs.}
    \label{fig:density_lower_bound}
\end{figure}

\begin{theorem}
  \label{thm:density_lower_bound}
  For infinitely many values of~$n$, there exists an $n$-vertex
  arc-RAC graph with $4.5n - O(\sqrt{n})$ edges.
\end{theorem}

\begin{proof}
  We first construct a tile and show that its drawing is
  indeed a valid \arac drawing.  Then it is easy to
  draw an embedded graph $G$ with the claimed edge density.

  Consider two circles $\alpha$ and $\beta$ that intersect in two
  points of different distance from the origin.
  Let~$X^-_{\alpha\beta}$ be the intersection point that is closer to
  the origin, and let~$X^+_{\alpha\beta}$ be the intersection point
  further from the origin.
  
  Let the
  vertices of the hexagon in a tile be $P_0 = X^+_{\alpha_0\alpha_1}$,
  $P_1 = X^-_{\beta_2\beta_0}$, $P_2 = X^+_{\alpha_1\alpha_2}$, $P_3 =
  X^-_{\beta_0\beta_1}$, $P_4 = X^+_{\alpha_2\alpha_0}$, and $P_5 =
  X^-_{\beta_1\beta_2}$.  Due to the symmetric definitions of the arcs,
  the angle between two consecutive vertices of the hexagon is
  $\pi/3$. Moreover, by a simple computation, we see that for
  each $j\in \{0, 1, 2\}$ and with $d = \sqrt{1/2} + \sqrt{1/6}$ being
  the distance of the vertices of the hexagon from the origin, we have:
  \[\begin{array}{l@{\;=\;}l@{\;=\;}l}
      P_{2j} & X^+_{\alpha_j\alpha_{j+1 \bmod 3}} &
      \big(d\cos(\frac{\pi}{2} + j \frac{2\pi}{3}),\,
      d\sin(\frac{\pi}{2} + j \frac{2\pi}{3})\big) \\
      P_{2j+3 \bmod 6} & X^-_{\beta_j\beta_{j+1 \bmod 3}} &
      \big(d\cos(\frac{\pi}{6} + (j+2) \frac{2\pi}{3}),\,
      d\sin(\frac{\pi}{6} + (j+2) \frac{2\pi}{3})\big).
    \end{array}\]   
  Thus, all the vertices of the hexagon are equidistant from its
  center, so the hexagon is regular.  According to
  Lemmas~\ref{lem:3circles} and~\ref{lem:3big_circles} all crossings
  of the arcs that belong to the same tile are orthogonal.  Now we
  argue that the arcs in $A$ and $B$ are contained in the regular
  hexagon.  To this end, we show that the arcs do not intersect the
  relative interior of the edges of the hexagon.  To see this, take, for
  example, the arc $\alpha_2$, which connects~$P_2$ and~$P_4$.  The
  line segment~$P_2P_4$ is orthogonal to the side~$P_1P_2$ of the
  hexagon.  As the center of~$\alpha_2$ is below~$P_2P_4$,
  the tangent of~$\alpha_2$ in~$P_2$ enters the interior of the
  hexagon in~$P_2$.  Thus, $\alpha_2$ does not intersect the relative
  interior of the edge $P_1P_2$ (or of any other edge) of the hexagon.
  Similarly we can show that the arcs in~$B$ do not intersect the
  relative interior of an edge of the hexagon.  Therefore, each tile
  is an \arac drawing, and $G$ is an \arac graph.

  Almost all vertices of the lattice with the exception of at most
  $O(\sqrt{n})$ vertices at the lattice's boundary have
  degree~9~\cite{afkmt-garacd-CGTA12}.  Hence $G$ has $4.5n -
  O(\sqrt{n})$ edges.
\end{proof}

As any $n$-vertex \rac{} graph has at most $4n-10$
edges~\cite{del-dgrac-TCS11}, we obtain the following.

\begin{corollary}
  The \arac graphs are a proper superclass of the \rac{0} graphs.
\end{corollary}

\section{Open Problems and Conjectures}
\label{sec:open-problems}

An obvious open problem is to tighten the bounds on the edge density
of \arac graphs in Theorems~\ref{thm:density_upper_bound}
and~\ref{thm:density_lower_bound}.  

Another immediate question is the
relation to \rac{1} graphs, which also extend the class of
\rac{0} graphs.  This is especially intriguing as the best known lower
bound for the maximum edge density of \rac{1} graphs is indeed larger than our
lower bound for \arac graphs whereas 
there may be \arac
graphs that are denser than the densest \rac{1} graphs.

The relation between \rac{k} graphs and 1-planar graphs is well understood
\cite{afkmt-garacd-CGTA12,bbhnr-nicpg-DAM17,bdlmm-oracdo1pg-TCS17,bdeklm-radicpg-TCS16,clwz-cdo1pgwracafb-CG19,el-racga1p-DAM13}.
What about the relation between \arac graphs and 1-planar graphs?
In particular, is there a 1-planar graph which is not \arac?

We are also interested in the area required by \arac drawings.
Are there \arac graphs that need exponential area to admit an \arac
drawing?  (A way to measure this off the grid is to consider the ratio
between the longest and the shortest edge in a drawing.)

Finally, the complexity of recognizing \arac graphs is open, but likely NP-hard.

\bibliographystyle{plainurl}
\bibliography{abbrv+refs}

\end{document}